\documentclass[a4paper,english,numberwithinsect,cleveref]{lipics-v2021}

\hideLIPIcs

\title{Shortest Dominating Set Reconfiguration under Token Sliding}

\author{Jan Maty{\'a}{\v s} K{\v r}i{\v s}{\v t}an\footnote{Corresponding author; kristja6@fit.cvut.cz}}{Faculty of Information Technology, Czech Technical University in Prague, Th{\'a}kurova 9, Prague, 160\,00, Czech Republic}{}{https://orcid.org/0000-0001-6657-0020}{}
\author{Jakub Svoboda}{Institute of Science and Technology, Austria}{}{https://orcid.org/0000-0002-1419-3267}{}

\authorrunning{J. M. K{\v r}i{\v s}{\v t}an and J. Svoboda}
\Copyright{Jan M. K{\v r}i{\v s}{\v t}an and Jakub Svoboda}
\keywords{reconfiguration, dominating set, trees, interval graphs, algorithms}
\funding{J. M. Křišťan acknowledges the support of the Czech Science Foundation Grant No. \mbox{22-19557S}. This work was supported by the Grant Agency of the Czech Technical University in Prague, grant No. SGS23/205/OHK3/3T/18. J. Svoboda acknowledges the support of the ERC CoG 863818 (ForM-SMArt) grant.}

\usepackage[utf8]{inputenc} 
\usepackage[nocompress]{cite}
\usepackage{lineno} 
\usepackage{todonotes}
\nolinenumbers%

\usepackage{amssymb} 
\usepackage{mathtools}
\usepackage{algorithm}
\usepackage{xspace}
\usepackage[noend]{algpseudocode}
\usepackage{thmtools} 
\usepackage{hyperref}
\usepackage{url}
\usepackage{float}

\usepackage{amsmath} 
\usepackage{thm-restate}
\usepackage{array}
\usepackage{tablefootnote}

{}
\crefname{reduction}{reduction}{reductions}

\graphicspath{{./images/}}

\makeatletter
\newcommand{\customlabel}[2]{%
\protected@write \@auxout {}{\string \newlabel {#1}{{#2}{\thepage}{#2}{#1}{}} }%
\hypertarget{#1}{#2}
}
\makeatother

\newcommand{\decprob}[3]{
  \vspace{2mm}
\noindent\fbox{
  \begin{minipage}{0.96\textwidth}
  \textsc{#1}\\
    \textbf{Input:} #2 \\
  \textbf{Output:} #3
  \end{minipage}
  }
  \vspace{2mm}
}

\newcommand{\calO}{\mathcal{O}}
\newcommand{\TS}{\textsc{Token Sliding}\xspace}

\newcommand{\probIS}{\textsc{Independent Set}\xspace}
\newcommand{\probDS}{\textsc{Dominating Set}\xspace}

\newcommand{\succM}[1]{\sigma_M(#1)\xspace}

\newcommand{\succr}[1]{\sigma(#1)\xspace}

\newcommand{\toki}{v}
\newcommand{\toky}{y}

\newcommand{\I}[1]{I(#1)}
\newcommand{\scr}{w\xspace}

\newcommand{\slide}[3]{#1(#2 \rightarrow #3)\xspace}

\newcommand{\reconfg}{\mathcal{R}}

\newcommand{\Supp}[1]{\text{Supp}(#1)}

\begin{document}

\maketitle

\begin{abstract}
In this paper, we present novel algorithms that efficiently compute a shortest reconfiguration sequence between two given dominating sets in trees and interval graphs under the \TS model.
In this problem, a graph is provided along with its two dominating sets, which can be imagined as tokens placed on vertices.
The objective is to find a shortest sequence of dominating sets that transforms one set into the other, with each set in the sequence resulting from sliding a single token in the previous set.
While identifying any sequence has been well studied, our work presents the first polynomial algorithms for this optimization variant in the context of dominating sets.
%
\end{abstract}

\section{Introduction}

Reconfiguration problems arise when the goal is to transform one feasible solution into another through a series of small steps, while ensuring that all intermediate solutions remain feasible.
These problems have been widely studied in the context of graph problems, such as \textsc{Independent Set}~\cite{Lokshtanov_2019, Demaine_2015, Bonamy_2017, Belmonte_2020, Bartier_2021, Bonsma_2014_claw_free}, \textsc{Dominating Set}~\cite{Bonamy_2021, Bousquet_2021, Lokshtanov_2018, Haddadan_2016, Bodlaender2021, Suzuki2014}, \textsc{Shortest Paths}~\cite{Kami_ski_2011, Gajjar_2022}, and \textsc{Coloring}~\cite{Bonsma_2009, Bonamy_2013, Bonsma_2014, Cereceda_2010}.
Reconfiguration problems have also been studied in the context of \textsc{Satisfiability}~\cite{Gopalan_2006, Mouawad_2017}.
See~\cite{Nishimura_2018} for a general survey.

In the case of the \probDS and other graph problems, the most commonly studied reconfiguration rules are \textsc{Token Jumping} and \textsc{Token Sliding}.
The feasible solution can be represented by tokens placed on the vertices of a graph.
Under \textsc{Token Jumping}, tokens can be moved one at a time to any other vertex, while under \textsc{Token Sliding}, tokens can only be moved one at a time to a neighboring vertex.

We focus on the \textsc{Token Sliding} variant, particularly on finding a shortest reconfiguration sequence.
This optimization variant has been extensively studied in the context of reconfiguring solutions for \textsc{Shortest Paths}~\cite{Kami_ski_2011}, \probIS~\cite{Yamada_2021, Hoang_2019}, and \textsc{Satisfiability}~\cite{Mouawad_2017}.

Our main contribution is the presentation of two polynomial algorithms for finding a shortest reconfiguration sequence between dominating sets on trees and on interval graphs.
This is achieved through a novel approach to finding a reconfiguration sequence, which we believe may have applications in the study of other related reconfiguration problems.

Bonamy et al.~\cite{Bonamy_2021} have shown that a reconfiguration sequence between dominating sets under \TS can be found in polynomial time when the input graph is a dually chordal graph, which is class of graphs encompassing trees and interval graphs.
We extend these results by demonstrating that finding a shortest reconfiguration sequence on trees and interval graphs can also be done in polynomial time.
A key observation is that we can match a simple lower bound on the length of the reconfiguration sequence.
We show that in case of dually chordal graphs, such lower bound cannot be matched on some instances and thus our techniques cannot be directly extended for that case.

We provide a brief overview of the known results, along with our contributions.
The lower bound for cases where the reachability problem is PSPACE-hard follows from the fact that the reconfiguration sequence must have superpolynomial length in some instances (unless PSPACE $=$ NP), as otherwise, a reconfiguration sequence would serve as a polynomial-sized proof of reachability.

\begin{table}[h!]
  \centering
  \begin{tabular}{|l|ll|}
    \hline
    Graph class & Decision problem & Optimization variant \\
    \hline
    Trees & P & $\calO(n)$~\textbf{(\Cref{cor:trees})}\\
    Interval graphs & P & $\calO(n^3)$~\textbf{(\Cref{thm:interval-graphs})}\\
    Dually chordal graphs & P & open\\
    Split & PSPACE-complete & $n^{\omega(1)}$\\
    Bipartite & PSPACE-complete & $n^{\omega(1)}$\\
    Planar & PSPACE-complete & $n^{\omega(1)}$\\
    \hline
  \end{tabular}
  \caption{
    Complexities of problems of reconfiguring dominating sets under \TS on various graph classes. The decision problem results are due to~\cite{Bonamy_2021}.
  }\label{tab:complexities}
\end{table}

\section{Preliminaries}
\subparagraph{Graphs and Trees}
Given a graph $G$ and vertex $v$, we denote the set of neighbors of $v$ by $N(v)$; moreover, $N[v] = N(v) \cup \{v\}$.
Given two vertices $v$ and $u$, we denote $d_G(v,u)$ as the distance between $v$ and $u$, that is the number of edges on a shortest path between $v$ and $u$.

Given a rooted tree $T$ rooted at $r$ and vertex $v$, we denote:
the subtree below $v$ as $T[v]$;
the depth of vertex $v$ as $d(v) = d(v, r)$;
the parent of $v$ as $p(v)$.

Let $\succr{u, v}$ be the set of vertices that follow $u$ on a shortest path from $u$ to $v$.
We assume that $\succr{u, u} = \emptyset$.

\subparagraph{Multisets}
Formally, a multiset $H$ of elements from a base set $S$ is defined as a \emph{multiplicity function} $H : S \rightarrow \mathbb{N} \cup \{0\}$.
We define the \emph{support} of $H$ as $\Supp{H} = \{ v \mid \text{$H(v) \geq 1$} \}$.
Let $H$ and $I$ be multisets, then
$H \cap I = \min(H, I)$,
$H \cup I = H + I$,
$H \setminus I = \max(H - I, 0)$,
$H \triangle I = (H \setminus I) \cup (I \setminus H)$.
The cardinality is defined as $|H| = \sum_{v \in S} H(v)$ and $v \in H$ if $v \in \Supp{H}$.
The Cartesian product $H \times I$ is a multiset of the base set $S \times S$ such that $(H \times I)((u, v)) = H(u) \cdot I(v)$ for all $u, v \in S$.

Note that if one of the operands is a set, we can assume that it is a multiset with multiplicities of 1 for all elements in the set.

\subparagraph{Graph problems}
Given a graph $G = (V,E)$, a set $D$ of vertices is dominating if every vertex is either in $D$ or a neighbor of a vertex in $D$.
A multiset $H$ is dominating if $\Supp{H}$ is dominating.
We say that given a set $S$ of vertices, the vertices with a neighbor in $S$ are \emph{dominated} from $S$.

For trees, we solve a more general problem called reconfiguration of hitting sets.
A \emph{hitting set} of a set system $\mathcal{S}$ is a set $H$ such that for each $S \in \mathcal{S}$ it holds $H \cap S \neq \emptyset$.
A multiset $H$ is a hitting set if $\Supp{H}$ is a hitting set.

\subparagraph{Reconfiguration sequence}
Given a graph $G$ a multiset $D$ of its vertices representing the placement of tokens, we denote $\slide{D}{u}{v} = (D \setminus \{u\}) \cup \{v\}$ the multiset resulting from jumping a token on $u$ to $v$ (or sliding a token on $u$ to $v$ if $\{u, v\} \in E(G)$).
Given a graph $G$ and a set $\Pi$ of feasible solutions, we say that a sequence of multisets $D_1, D_2, \dots, D_\ell$ (of length $\ell$) is a \emph{reconfiguration sequence} under \TS between $D_1, D_\ell \in \Pi$ if
\begin{itemize}
  \item $D_i \in \Pi$ for all $1 \leq i \leq \ell$,
  \item $D_{i+1}
  = \slide{D_i}{u}{v}$ such that $v \in V(G)$, $u \in D_i$ and $\{u, v\} \in E(G)$ for all $1 \leq i < \ell$.
\end{itemize}

The sequence can be concisely represented by a sequence of \emph{moves}.
Given a starting multiset $D_s$, moves $(u_1, v_1), \dots, (u_{k-1}, v_{k-1})$ induce sequence $D_1, D_2, \dots, D_k$ such that $D_1 = D_s, D_{i+1} = \slide{D_i}{u_i}{v_i}$ for all $1 \leq i < k$.
This allows us to formally give the main problem of this paper.


\decprob{Shortest reconfiguration of dominating sets under \TS}
{Graph $G = (V, E)$ and two dominating sets $D_s$ and $D_t$.}
{Shortest sequence of moves $(u_1, v_1), \dots, (u_{k-1}, v_{k-1})$ inducing a reconfiguration sequence under \TS between $D_s$ and $D_t$.}

In the case of trees, we design an algorithm that finds a reconfiguration sequence whenever the feasible solutions can be expressed as hitting sets of
a set system $\mathcal{S}$ such that every $S \in \mathcal{S}$ induces a subtree of the input tree $T$.
Several problems can be formulated in terms of such hitting sets.
\begin{itemize}
  \item If $\mathcal{S}$ is the set of all closed neighborhoods of $T$, then the hitting sets of $\mathcal{S}$ are exactly the dominating sets of $T$.
  \item If $\mathcal{S}$ is the set of all edges, then the hitting sets are exactly all vertex covers of $T$.
  \item An instance of an (unrestricted) vertex multicut is equivalent to a hitting set problem with $\mathcal{S}$ being the set of all paths which must be cut.
\end{itemize}
The general problem of reconfiguring hitting sets is as follows.

\decprob{Shortest reconfiguration of hitting sets under \TS}
{Graph $G = (V, E)$ and two hitting sets $H_s$ and $H_t$ of a set system $\mathcal{S} \subseteq 2^{V(T)}$.}
{Shortest sequence of moves $(u_1, v_1), \dots, (u_{k-1}, v_{k-1})$ inducing a reconfiguration sequence under \TS between $H_s$ and $H_t$.}

\subparagraph{Reconfiguration graph}
Given a graph $G$ and an integer $k$, the \emph{reconfiguration graph} $\mathcal{R}(G,k)$ has as vertices all feasible solutions, in our case dominating multisets, of size $k$.
Two vertices are adjacent whenever one can be reached from the other in a single move, i.e. sliding a token.
Note that the shortest reconfiguration of dominating sets under \TS between $D_s$ and $D_t$ is equivalent to finding a shortest path in $\mathcal{R}(G, |D_s|)$ between $D_s$ and $D_t$.
Furthermore, as each move under \TS is reversible, the edges of $\mathcal{R}(G, k)$ are undirected.
Thus, finding a shortest path from $D_s$ to $D_t$ is equivalent to finding a shortest path from $D_t$ to $D_s$.

It follows that for $D_s$ and $D_t$, if $D_s \neq D_t$ and both are in the same connected component of $\reconfg(G, |D_s|)$, then $d_{\reconfg(G, |D_s|)}(D_s, D_t)$ is the minimum number of moves inducing a reconfiguration sequence between $D_s$ and $D_t$.
If $G$ and $|D_s|$ is clear from the context, we consider $\reconfg = \reconfg(G, |D_s|)$.

\subparagraph{Interval graphs}

A graph $G$ is an \emph{interval graph} if each vertex $v$ can be mapped to a different closed interval $I(v)$ on the real line so that ${v, u} \in E$ if and only $I(v) \cap I(u) \neq \emptyset$.
Such a mapping to intervals is called \emph{interval representation}.

We denote the endpoints of an interval $I(v)$ as $\ell(v)$ and $r(v)$ so that $I(v) = [\ell(v), r(v)]$.
It is known that every interval graph has an interval representation with integer endpoints in which no two endpoints coincide.
We assume that is the case throughout this paper, as such an interval representation can be computed in linear time~\cite{Corneil2010}.

We say that interval $I$ is \emph{to the left} of $J$ (or that $J$ is \emph{to the right} of $I$) if $r(I) < \ell(J)$.
Similarly, we say that $I$ is \emph{nested} in $J$ (or that $J$ \emph{contains} $I$) if $\ell(J) < \ell(I)$, $r(I) < r(J)$.
Furthermore, we say that $I$ \emph{left-intersects} $J$ (or that $J$ \emph{right-intersects} $I$) if $\ell(I) < \ell(J) < r(I) < r(J)$.
Note that every pair of intervals is in exactly one of those relations.
We say that two vertices $u$ and $v$ of an interval graph are in a given relationship if their intervals $I(u)$ and $I(v)$ in a fixed interval representation are in the given relationship.
\section{Lower bounds on lengths of reconfiguration sequences}\label{sec:lb}

We can obtain a lower bound on the length of a reconfiguration sequence by dropping the requirement that the tokens induce a feasible solution (such as a dominating set) at each step.
The problem of finding such shortest reconfiguration sequence is polynomial-time solvable by reducing to the minimum-cost matching in bipartite graphs.

Let $G$ be a graph, $D_s, D_t \subseteq V(G)$ be the multisets representing tokens.
Then $M \subseteq D_s \times D_t$ is a \emph{matching} between $D_s$ and $D_t$ if for every $v \in D_s$, there is exactly $D_s(v)$ pairs $(v, \cdot) \in M$ and similarly for every $v \in D_t$, there is exactly $D_t(v)$ pairs $(\cdot, v) \in M$.
Note that $M$ is a multiset and the same pair may be contained in $M$ multiple times.

We say that $u \in D_s$ and $v \in D_t$ are matched in $M$ if $(u, v) \in M$.
We also use $M(u)$ to denote the set of matches of $u$, that is the vertices $v$ such that $(u, v) \in M$.
The \emph{cost} $c(M)$ of the matching $M$ is defined as
\[
  c(M) = \sum_{(u, v) \in M} d_G(u, v) \cdot M(u, v).
\]
We say that a matching has \emph{minimum cost} if its cost is the minimum over all possible matchings between $D_s$ and $D_t$ and denote this cost as $c^*(D_s, D_t)$.
We use $\succM{u}$ to denote the vertices which follow $u$ on some shortest path to some match $M(u) \neq u$.
Formally \[\succM{u} = \bigcup_{v \neq u : (u, v) \in M} \succr{u, v}.\]
We define $M^{-1}$ so that $M^{-1}(v, u) = M(u, v)$ for all $u \in D_s, v \in D_t$.


\begin{lemma}\label{lem:lb}
  Every sequence of moves inducing a reconfiguration sequence between $D_s$ and $D_t$ under \TS has length at least $c^*(D_s, D_t)$.
\end{lemma}
\begin{proof}
  Suppose a reconfiguration sequence using fewer than $c^*(D_s, D_t)$ moves exists.
  Let $M$ be a matching between $D_s$ and $D_t$ of minimum cost.
  Then we can track the moves of each token and construct a matching $M'$ between $D_s$ and $D_t$ given by the starting and ending position of each token.
  Note that the cost of each matched pair is at most the length of the path travelled by the given token.
  Thus in total the cost of $M'$ is at most the total number of moves used.
  Hence, we have $c(M', D_s, D_t) < c^*(D_s, D_t)$, a contradiction.
\end{proof}

The following observation shows that in a minimum-cost matching, if a token can be matched with zero cost, we can assume that is the case for all such tokens.

\begin{lemma}\label{obs:nice-matching}
  For graph $G$, let $D_s, D_t$ be multisets of the same size and let $I = D_s \cap D_t$.
  Then there exists minimum-cost matching $M$ in $G$ between $D_s$ and $D_t$ such that for every $v \in I$ we have $M(v,v) = I(v)$.
\end{lemma}
\begin{proof}
  Given a minimum matching $M$ between $D_s$ and $D_t$ and $v$ such that $M(v, v) < I(v)$, we show that we can produce $M'$ of the same cost such that $\sum_{u \in I} M'(u, u) > \sum_{u \in I} M(u, u)$.
  Note that there exist $(x, v), (v, y) \in M$ with $x \neq v, y \neq v$ as otherwise $M( (v, v) ) = I(v)$.
  Then we define $M' = (M \setminus \{(x,v),(v,y)\}) \cup \{ (v,v), (x,y)\}$.

  We have $c(M') - c(M) = - d_G(x, v) - d_G(v, y) + d_G(v, v) + d_G(x, y) = -d_G(x, v) - d_G(v, y) + d_G(x, y)$.
  From the triangle inequality $d_G(x,y) \le d_G(x,v) + d_G(v,y)$, thus we have that the cost of $M'$ is at most the cost of $M$ and thus is minimum.
  By repeated application, we arrive at minimum-cost matching $M^*$ with $M^*(v, v) = I(v)$ for all $v$.
\end{proof}

The following observation shows that, given $D_s$ and $D_t$, if we pick a token in $D_s$ and slide it along an edge to decrease its distance to its match in a minimum-cost matching, the resulting $D'_s$ and $D_t$ have minimum cost of matching of exactly one less than $D_s$ and $D_t$.
Thus if each move in the reconfiguration sequence is of such a kind, the length of the resulting sequence will match the lower bound of \Cref{lem:lb}.

\begin{lemma}\label{obs:matching-move}
  Let $M^*$ be a minimum-cost matching between $D_s$ and $D_t, (u, g) \in M^*$ and $v \in \succr{u, g}$ a vertex that follows $u$ on a shortest path from $u$ to $g$.
  Furthermore, let
  \[
    M = \bigl( M^* \setminus \{ (u, g) \} \bigr) \cup \{(v, g)\}.
  \]
  Then $M$ is a minimum-cost matching between $\slide{D_s}{u}{v}$ and $D_t$.
  Furthermore, $c^*(D_s, D_t) = c^*(\slide{D_s}{u}{v}, D_t) + 1$.
\end{lemma}
\begin{proof}
  From definition $c(M^*) - c(M) = d_G(u,g) - d_G(v,g)$, but $v$ is the vertex on the path from $u$ to $g$,
  so $d_G(u,g) - d_G(v,g) = 1$ and $c^*(\slide{D_s}{u}{v}, D_t) \le c^*(D_s,D_t) - 1$.

  Suppose that $c^*(\slide{D_s}{u}{v}, D_t) < c^*(D_s,D_t) - 1$, i.e., there exists matching $M'$ between $\slide{D_s}{u}{v}$ and $D_t$ such that $c(M') < c(M)$.
  From $M'$, we construct a matching $M''$ between $D_s$ and $D_t$ such that $c(M'') < c(M^*)$, which is a contradiction.

  Let $x \in M'(v)$ and set $M'' = (M' \setminus \{(v,x)\}) \cup \{(u,x)\}$.
  The cost $c(M'') \le c(M') + 1$, since the distance between $v$ and $u$ is $1$.
  That means if $c(M') < c(M)$, then $c^*(D_s,D_t) < c(M^*)$, but $M^*$ is minimum-cost matching.

  Therefore, $M$ is a minimum-cost matching between $\slide{D_s}{u}{v}$ and $D_t$ and $c^*(D_s, D_t) = c^*(\slide{D_s}{u}{v}, D_t) + 1$.
\end{proof}

\section{Algorithms for finding a shortest reconfiguration sequence}

In the following sections, we present algorithms for finding a shortest reconfiguration sequence between dominating sets on trees and interval graphs under \TS.


\subsection{Trees}
We present an algorithm that, given a tree $T$ and two hitting sets $H_s, H_t$ of a set system $\mathcal{S}$ such that every $S \in \mathcal{S}$ induces a subtree of $T$, finds a shortest reconfiguration sequence between $H_s$ and $H_t$ under \TS.
As dominating sets are exactly the hitting sets of closed neighborhoods, the algorithm finds a shortest reconfiguration sequence between two dominating sets.
Note that $\mathcal{S}$ need not be provided on the input.

Consider the reconfiguration graph $\mathcal{R}(G, |H_s|)$, whose vertices are the all the hitting multisets of $\mathcal{S}$ of size $|H_s|$.
The high-level idea is to extend two paths in $\mathcal{R}(G, |H_s|)$, one from $H_s$ and another from $H_t$, until they reach a common configuration.
We repeatedly identify a subtree $T[v]$ of the rooted $T$ for which the configurations $H_s$ and $H_t$ are identical, except for $v$ itself.
Then, we modify either $H_s$ or $H_t$ by sliding the token (or tokens) on $v$ to its parent, ensuring that $H_s$ and $H_t$ become equal when restricted to $T[v]$.

\Cref{alg:trees} describes the algorithm.
We assume that the input tree is rooted in some vertex $r$.


\begin{algorithm}
  \caption{Reconfiguration of hitting sets in trees}
  \begin{algorithmic}[1]
    \Procedure{ReconfTree}{$T, H_s, H_t$}\label{alg:trees}
    \If {$H_s = H_t$}
      \Return $\emptyset$
    \EndIf
    \State $v \gets$ vertex $v$ such that $H_s(v) \neq H_t(v)$ and $H_s(u) = H_t(u)$ for all $u \in T(v)$.
    \If {$H_s(v) > H_t(v)$}
      \State \Return $(v, p(v))$ + \Call{ReconfTree}{$T, \slide{H_s}{v}{p(v)}, H_t$}
    \Else
      \State \Return \Call{ReconfTree}{$T, H_s, \slide{H_t}{v}{p(v)}$} + $(p(v), v)$
    \EndIf
    \EndProcedure
  \end{algorithmic}
\end{algorithm}


The proof of correctness uses techniques of \Cref{sec:lb}.
While the correctness of the algorithm can be proved without them, we believe this presentation is helpful for understanding the proofs in subsequent sections.

\begin{theorem}\label{thm:tree-reconf}
  Let $T$ be a tree on $n$ vertices and $H_s$ and $H_t$ hitting sets of a set system $\mathcal{S}$ in which every $S \in \mathcal{S}$ induces a subtree of $T$.
  Then \Call{ReconfTree}{} (\Cref{alg:trees}) correctly computes a solution to \textsc{Shortest reconfiguration of hitting sets under \TS}.
  Furthermore, it runs in time $\mathcal{O}(n)$.
\end{theorem}
\begin{proof}
  \begin{figure}
    \centering
    \subfloat[]{
      \includegraphics[width=0.25\textwidth, page=1]{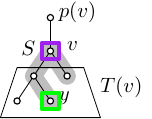}
      \label{fig:tree1}
    }
    \subfloat[]{
      \includegraphics[width=0.25\textwidth, page=2]{images/tree}
      \label{fig:tree2}
    }
    \caption{Illustrations accompanying the proof of \Cref{thm:tree-reconf}.
    The green squares denote tokens of $H_s$, the purple squares denote tokens of $H_t$.
    The grey areas show examples of $S$ in the two considered cases.}
    \label{fig:tree}
  \end{figure}
  We will show that \Call{ReconfTree}{} outputs a sequence of $d_\mathcal{R}(H_s, H_t)$ moves which induces a reconfiguration sequence between the two hitting sets $H_s, H_t$ of $\mathcal{S}$.
  If $H_s = H_t$, then $d_\mathcal{R}(H_s, H_t) = 0$ and the procedure correctly outputs an empty sequence.
  Thus assume that $H_s \neq H_t$.

    Suppose that $T$ is rooted in $r$ and let $v$ be a vertex such that $H_s(v ) \neq H_t(v)$ and $H_s(u) = H_t(u)$ for all $u \in T(v)$.
    Without loss of generality, assume that $H_s(v) > H_t(v)$ as otherwise, we can swap $H_s$ and $H_t$.

  \begin{claim}
    $\slide{H_s}{v}{p(v)}$ is a hitting set of $\mathcal{S}$.
  \end{claim}
  \begin{claimproof}
  Suppose that $H'_s = \slide{H_s}{v}{p(v)}$ is not a hitting set of $\mathcal{S}$.
    It follows that $\Supp{H_s}~\nsubseteq~\Supp{H'_s}$ and therefore $H_s(v) = 1$ and $H'_s(v) = 0$ and $H'_s$ is not intersecting only sets $S \in \mathcal{S}$ such that $v \in S$ and $p(v) \notin S$.
    Furthermore, $H_t(v) = 0$ as $H_t(v) < H_s(v)$.

  Let $S \in \mathcal{S}$ be a set not intersecting $H'_s$ and let $y \in S \cap H_t$.
  Such $y$ distinct from $v$ must exist as $H_t$ is a hitting set of $\mathcal{S}$ and $v \notin H_t$.
    If $y \in T[v]$, then $y \in H_s$ as $H_s(y) = H_t(y)$ by the choice of $v$, which contradicts $H'_s$ not intersecting $S$.
  This case is shown in \Cref{fig:tree1}.

  Therefore $y \in T \setminus T[v]$.
  Note that the path connecting $v$ with $y$ must visit $p(v)$, thus as $S$ induces a subtree and contains $u$ and $y$, it contains $p(v)$ as well and therefore $S$ intersects $H'_s$.
  This case is shown in \Cref{fig:tree2}.
  \end{claimproof}

\begin{claim}
  The number of moves outputted by \Call{ReconfTree}{$T, H_s, H_t$} is equal to $d_\mathcal{R}(H_s, H_t)$.
\end{claim}
\begin{claimproof}
  We first claim that if $H_s, H_t$ are two hitting sets of $\mathcal{S}$ with the same size, then $d_{\mathcal{R}}(H_s, H_t) = c^*(H_s, H_t)$.
  Furthermore, we show that a move from $v$ to $p(v)$ decreases the cost of a minimum-cost matching between $H_s$ and $H_t$ by one, where $v$ is a vertex such that $H_s(v ) \neq H_t(v)$ and $H_s(u) = H_t(u)$ for all $u \in T(v)$.
  This together implies that each outputted move decreases the distance in the reconfiguration graph by one.

  We prove the claim by induction on $c^*(H_s, H_t)$ that $d_{\mathcal{R}}(H_s, H_t) = c^*(H_s, H_t)$ for any hitting sets $H_s, H_t$ of the same size.
  First note that $c^*(H_s, H_t) = 0$ if and only if $H_s = H_t$.
  Now, suppose that $c^*(H_s, H_t) \geq 1$.

  Let $M^*$ be a minimum-cost matching between $H_s$ and $H_t$ such that tokens with distance $0$ are matched to each other, such matching exists by \Cref{obs:nice-matching}.

  Let $H'_s = \slide{H_s}{v}{p(v)}$.
  As all tokens in $T(v)$ are matched by $M^*$ only to the same vertex, it holds $M^*(v) \subseteq V \setminus T[v]$.
  Therefore $p(v)$ is the next vertex on the path from $v$ to some $g \in M^*(v)$ and thus by \Cref{obs:matching-move} it holds $c^*(H_s', H_t) = c^*(H_s, H_t) - 1$.
  As $H'_s$ is a hitting set of $\mathcal{S}$ by the previous claim, it follows from the induction hypothesis that $d_\mathcal{R}(H'_s, H_t) = c^*(H_s', H_t)$.
  Now, note that $d_\mathcal{R}(H_s, H_t) \geq c^*(H_s, H_t)$ by \Cref{lem:lb}.
  On the other hand,
  \[
    d_\mathcal{R}(H_s, H_t) \leq d_\mathcal{R}(H'_s, H_t) + 1 = c^*(H'_s, H_t) + 1 = c^*(H_s, H_t)
  \] as $H_s$ can be reached from $H'_s$ by a single token slide.
  This concludes the proof of the inductive step.

  As $d_\mathcal{R}(H'_s, H_t) = d_\mathcal{R}(H_s, H_t) - 1$, each call of the algorithm decreases the distance between the hitting sets by one and also outputs one move.
  Thus the resulting reconfiguration sequence is shortest possible.
\end{claimproof}

We now describe how to implement \Cref{alg:trees} so that it achieves the linear running time.
Note that we assume that the input $H_s$ and $H_t$ of the initial call of $\Call{ReconfTree}{}$ are subsets of $V(T)$ and therefore $|H_s|, |H_t| \leq n$.
Next, we show how to compress the output to $\mathcal{O}(n)$ size.
Whenever $|H_s(v) - H_t(v)| > 1$, we can perform all $|H_s(v) - H_t(v)|$ moves from $v$ to $p(v)$ at once and output them as a triple $(v, p(v), |H_s(v) - H_t(v)|)$ if $H_s(v) > H_t(v)$ or $(p(v), v, |H_s(v) - H_t(v)|)$ in case $H_s(v) < H_t(v)$.

Note that with this optimization, the vertex $v$ on line $3$ is distinct for each call of $\Call{ReconfTree}{}$.
Furthermore, we can fix in advance the order in which we pick candidates of $v$ on line $3$ by ordering the vertices of $T$ by their distance from $r$ in decreasing order.
This is correct as the depth of the lowest vertex satisfying the condition of line $3$ cannot increase in the subsequent calls.
Then, the process of finding $v$ on line $3$ has total runtime of $\mathcal{O}(n)$ over the course of the whole algorithm.
\end{proof}

\begin{corollary}\label{cor:trees}
  Let $T$ be a tree on $n$ vertices and $D_s, D_t$ dominating sets of $T$ such that $|D_s| = |D_t|$.
  \Cref{alg:trees} finds a shortest reconfiguration sequence between $D_s$ and $D_t$ under \TS in $\mathcal{O}(n)$ time.
\end{corollary}

In general, the length of the reconfiguration sequence can be up to $\Omega(n^2)$,
for instance when $\Omega(n)$ tokens are required to move from one end of a path to the other end, as each must move to a distance of at least $\Omega(n)$.
However, when this happens, a lot of tokens move by one edge and we can move them at the same time, so the running time of the algorithm can be smaller than the number of moves.



\subsection{Interval graphs}

In this section, we describe a polynomial-time algorithm for finding a shortest reconfiguration sequence between two dominating sets under the \TS model in interval graphs.
As with trees, we demonstrate that the distance between two dominating sets in interval graphs is equal to the lower bound established in \Cref{lem:lb}.
Our approach involves a minimum-cost matching between the dominating sets $D_s$ and $D_t$ to identify a valid move.
The key insight of this algorithm is that we can always recalculate the minimum-cost matching to enable sliding at least one token along a shortest path towards its corresponding match.

The following pseudocode outlines the algorithm.
A minimum-cost matching $M$ between $D_s$ and $D_t$ is assumed to be provided on the input.

\begin{algorithm}[H]
  \caption{Reconfiguration of dominating sets in interval graphs}\label{alg:intervals}
  \begin{algorithmic}[1]
    \Procedure{ReconfIG}{$G, D_s, D_t, M$}
    \If {$D_s = D_t$}
    \Return $\emptyset$
    \EndIf
    \If {$ \exists (u, v) \in M, u' \in \succr{u, v}$ such that $\slide{D_s}{u}{u'}$ is dominating}
      \State \Return $(u \rightarrow u')$ + \Call{ReconfIG}{$G, \slide{D_s}{u}{u'}, D_t$}
    \EndIf
    \If {$ \exists (u, v) \in M, v' \in \succr{v, u}$ such that $\slide{D_t}{v}{v'}$ is dominating}
      \State \Return \Call{ReconfIG}{$G, D_s, \slide{D_t}{v}{v'}$} + $(v' \rightarrow v)$
    \EndIf
    \State $M' \gets$ \Call{FixMatching}{$G, D_s, D_t, M$}\label{line:fix-matching}
    \State \Return \Call{ReconfIG}{$G, D_s, D_t, M'$}
    \EndProcedure

    \Procedure{FixMatching}{$G, D_s, D_t, M$}
    \State $v \in D_s \mathbin{\triangle} D_t$ with minimum possible $r(v)$.
    \If {$v \in D_t$}
      \State \Return \Call{FixMatching}{$G, D_t, D_s, M^{-1}$}$^{-1}$ \Comment{Symmetric solution, swap $D_s, D_t$}
    \EndIf
    \State Find $y \in D_t \setminus D_s, y' \in M(y), v' \in M(v)$ such that $D'_s = \slide{D_s}{v}{y}$ is dominating and\label{line:find-y}
    $M' = M \setminus \{(v, v'), (y', y)\} \cup \{(v, y), (y', v')\}$ is a minimum-cost matching\\between $D_s$ and $D_t$.
    \State \Return $M'$
    \EndProcedure
  \end{algorithmic}
\end{algorithm}

The bulk of the proof consists of showing that the procedure \textsc{FixMatching} is correct, in particular that the call on line \ref{line:find-y} succeeds.
First, we present a technical observation related to shortest paths in interval graphs.

\begin{observation}\label{obs:sp}
  Let $P = (v_1, v_2, \dots, v_k)$ be a shortest path between $v_1$ and $v_k$ in an interval graph with $r(v_1) < r(v_k)$ and $k \geq 3$.
  Then $v_{i+1}$ right-intersects $v_i$ and $v_{i+2}$ does not intersect $v_i$ for all $i \in \{1, \dots, k - 2\}$.
\end{observation}
\begin{proof}
  If for some $i \in \{1, \dots, k - 2\}$ $v_{i+2}$ intersects $v_i$, then we can create a shorter path from $v_1$ to $v_k$ by removing $v_{i+1}$ from $P$, contradicting $P$ being a shortest path.

  Suppose that for some $i \in \{1, \dots, k - 2\}$ it holds $r(v_{i+1}) < r(v_{i})$.
  Note that a shortest path contains no nested intervals with a possible exception of $v_1$ and $v_k$, as every other nested interval can be removed to make the path shorter.
  Thus $v_{i+1}$ left-intersects $v_i$.
  Let $v_j$ be the first next vertex after $v_i$ such that $r(v_i) < r(v_j)$.
  If none such exists, then $v_i$ must intersect $v_k$ and thus the path can be made shorter.
  Otherwise we show that $v_j$ intersects $v_i$.
  If it does not, then $\ell(v_j) > r(v_i)$.
  But for the path to be connected, another interval $v_a$ must cover $[r(v_i), \ell(v_j)]$.
  Such interval either has $r(v_j) < r(v_a)$, thus $v_j$ is nested in $v_a$ or $r(v_i) < r(v_a) < r(v_j)$, contradicting the choice of $v_j$.
\end{proof}

The following lemma shows that we can efficiently recompute the minimum-cost matching to ensure that for some token a valid move across a shortest path to its match will be available.

\begin{lemma}\label{lem:move}
  The call of \Call{FixMatching}{} on line \ref{line:fix-matching} returns a minimum-cost matching $M'$ between $D_s$ and $D_t$ such that at least one the following holds.
  \begin{itemize}
    \item There is $(u, v) \in M', u' \in \succr{u, v}$ such that $\slide{D_s}{u}{u'}$ is dominating,
    \item there is $(u, v) \in M', v' \in \succr{v, u}$ such that $\slide{D_t}{v}{v'}$ is dominating.
  \end{itemize}
\end{lemma}
\begin{proof}
  \begin{figure}
    \centering
    \begin{subfigure}{.30\textwidth}
      \centering
      \includegraphics[page=1]{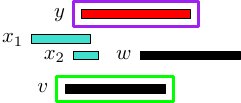}
      \caption{}\label{fig:intervals}
    \end{subfigure}
    \begin{subfigure}{.30\textwidth}
      \centering
      \includegraphics[page=2]{images/interval_graphs}
      \caption{}\label{fig:move2}
    \end{subfigure}
    \begin{subfigure}{.30\textwidth}
      \centering
      \includegraphics[page=3]{images/interval_graphs}
      \caption{}\label{fig:walk}
    \end{subfigure}
    \caption{Illustrations accompanying the proof of \Cref{lem:move}. The green squares denote tokens of $D_s$, the purple squares denote tokens of $D_t$.}
  \end{figure}
  The idea of the proof is in showing that if no token can move along a shortest path to its match, then there is always a way to modify the matching which does not increase cost and makes moving along a shortest path possible for at least one token.
  In particular, we need to show that the operation of finding $y$ on line \ref{line:find-y} always succeeds and the constructed $M'$ is a minimum-cost matching between $D_s$ and $D_t$.

  As the algorithm has not finished on line 2, it holds $D_s \neq D_t$.
  Let $M$ be a minimum-cost matching between $D_s$ and $D_t$.
  If for some $(u, v) \in M, \scr \in \succr{u, v}, \scr' \in \succr{v, u}$ $\slide{D_s}{u}{\scr}$ or $\slide{D_t}{v}{\scr'}$ is dominating, then we would not have reached line~\ref{line:fix-matching}.
  Therefore, assume that for every $(u, v) \in M, \scr \in \succr{u, v}, \scr' \in \succr{v, u}$ neither $\slide{D_s}{u}{\scr}$ nor $\slide{D_t}{v}{\scr'}$ is dominating.

  Let $(\toki, \toki') \in M$ such that $\toki \neq \toki'$ and $\min(r(\toki), r(\toki'))$ is minimum possible.
  Without loss of generality, assume that $r(\toki) < r(\toki')$ as otherwise, we can swap $D_s$ and $D_t$.

  \begin{claim}
    For every $\scr \in \succM{\toki}$, $\I{\scr}$ right-intersects $\I{\toki}$.
  \end{claim}
  \begin{claimproof}
    Suppose that $\I{\scr}$ contains $\I{\toki}$.
    Then $\slide{D_s}{\toki}{\scr}$ is dominating, a contradiction.
    Now suppose that $\I{\toki}$ contains $\I{\scr}$.
    Then by \Cref{obs:sp} it holds $(\toki, \scr) \in M$, which implies that $\toki \in \succr{\scr, \toki}$ and $\slide{D_t}{\scr}{\toki}$ is dominating as $N[\toki] \subseteq N[\scr]$, a contradiction.

    The remaining case is that $\I{\scr}$ left-intersects $\I{\toki}$.
    Then again, by \Cref{obs:sp} it holds $(\toki, \scr) \in M$ and this contradicts the choice of $\toki$ as $r(\scr) < r(\toki)$.
  \end{claimproof}

  Now, let $\scr \in \succM{\toki}$ be a fixed vertex and consider why $\slide{D_s}{\toki}{\scr} = D_s'$ is not dominating.
  Let $x_1, \dots, x_k \subset N(\toki) \setminus N(\scr)$ be the vertices that are not dominated by $D_s'$.
  \begin{claim}
    There exists $\toky \in N(\toki) \cap (D_t \setminus D_s)$ such that all $x_i$ are adjacent to $\toky$.
  \end{claim}
  \begin{claimproof}
  First, we will show that $\I{x_i}$ is to the left of $\I{\scr}$ for all $x_i$.
  Note that as each no $x_i$ is adjacent to $\scr$, $\I{x_i}$ is either to the left or to the right of $\I{\scr}$.

  Suppose there is some $\I{x_i}$ to the left of $\I{\scr}$ and some $\I{x_j}$ to the right of $\I{\scr}$, then $\I{\toki}$ contains $\I{\scr}$, which as previously argued may not be the case.
  The remaining case is that all $\I{x_i}$ are to the right of $\I{\scr}$, which would imply that $\I{\scr}$ left-intersects $\I{\toki}$, which again was shown not to hold.
  Therefore, each $\I{x_i}$ is to the left of $\I{\scr}$.
  This further implies that $\ell(\toki) < r(x_i) < \ell(\scr)$, thus each $\I{x_i}$ is either nested in $\I{\toki}$ or left-intersects $\I{\toki}$.
  
  Observe that each $x_i$ is adjacent to some $y_i \in D_t \setminus D_s$ and $r(\toki) < r(y_i)$ by the choice of $\toki$.
  Therefore, there exists $\toky \in D_t \setminus D_s$ such that $I(\toky)$ contains $\min(r(x_1), \dots, r(x_k))$.
  Together, we get
  \begin{equation}
    \ell(\toky) < \min(r(x_1), \dots, r(x_k)) \leq \max(r(x_1), \dots, r(x_k)) < \ell(\scr) < r(\toki) < r(\toky)
    \label{eq:ordering}
  \end{equation}
  and therefore $\toky$ is adjacent to all $x_i$.
  See \Cref{fig:intervals} for an illustration.
  As $\ell(\toky) < r(\toki) < r(\toky)$, $\I{\toky}$ either right-intersects $\I{\toki}$ or contains $\I{\toki}$ and thus $\toki$ and $\toky$ are adjacent.
  \end{claimproof}

  The rest of the proof consists of two claims.
  The first is that $\slide{D_s}{\toki}{\toky}$ is dominating.
  The second is that $(\toki, \toky) \in M'$ for some minimum-cost matching $M'$ between $D_s$ and $D_t$.

  \begin{claim}
    $\slide{D_s}{\toki}{\toky}$ is dominating.
  \end{claim}
  \begin{claimproof}
    Let $D' = \slide{D_s}{\toki}{\toky}$.
    If $\toky$ contains $\toki$, then $N[\toki] \subseteq N[\toky]$, therefore $D_s \subseteq D'_s$ and $D'_s$ is dominating.
    Thus assume that $\toky$ right-intersects $\toki$, which is the only remaining case as shown above.

    Suppose $u \in N(\toki)$ is a vertex which is not dominated from $D_s'$.
    Note that $u$ must not be adjacent to $\toky$ and at the same time be adjacent to $\toki$, therefore $u$ is to the left of $\toky$.
    Then $u$ is to the left of all $\scr \in \succM{\toki}$ as $\ell(\toky) < \ell(\scr)$, thus $u$ is not dominated in $\slide{D_s}{\toki}{w}$ and therefore $u = x_i$ for some $i$.
    This implies that $u$ is not adjacent to $\toky$ and this contradicts the choice of $\toky$.
  \end{claimproof}

    Let $\toki' \in D_t$ such that $\toki' \neq \toki$ and $(\toki, \toki') \in M$.
    Similarly, let $\toky' \in D_s$ such that $\toky' \neq \toky$ and $(\toky', \toky) \in M$.
    We define the new matching $M'$ as
    \[
      M' = \bigl(M \setminus \{(\toki, \toki'), (\toky', \toky)\}\bigr) \cup \{(\toki, \toky), (\toky', \toki')\}.
    \]
  \begin{claim}
    $M'$ is a minimum-cost matching between $D_s$ and $D_t$.
  \end{claim}
  \begin{claimproof}
    We prove that $c(M') \leq c(M)$.
    Given that $d(\toki, \toky) = 1$ it suffices to show that
    \begin{align*}
      d(\toki, \toky) + d(\toki', \toky') &\leq d(\toki, \toki') + d(\toky, \toky') \\
      d(\toki', \toky') &\leq d(\toki, \toki') + d(\toky, \toky') - 1.
    \end{align*}
    Let $\scr_\toki \in \succr{\toki, \toki'}$ and $\scr_\toky \in \succr{\toky', \toky}$.

    \subparagraph{Case 1: $\scr_\toki$ and $\scr_\toky$ are adjacent.}
    We can construct a walk $W$ from $\toki'$ to $\toky'$ by concatenating shortest paths between each two consecutive vertices in $(\toki',\scr_\toki,\scr_\toky,\toky')$.
    It holds that $d(\toki', \toky')$ is at most the number of edges of $W$ and therefore
    \begin{align*}
      d(\toki', \toky') &\leq d(\toki', \scr_\toki) + d(\scr_\toki, \scr_\toky) + d(\scr_\toky, \toky') \\
      &= d(\toki', \toki) - 1 + 1 + d(\toky, \toky') - 1 \\
      &= d(\toki, \toki') + d(\toky, \toky') - 1.
    \end{align*}
    \subparagraph{Case 2: $\scr_\toki$ and $\scr_\toky$ are not adjacent and $I(\scr_\toki)$ is nested in $I(\toky)$.}
    Suppose that $\toki' = \scr_\toki$.
    Given that $I(\scr_\toki)$ is nested in $I(\toky)$, it follows that $N[\toki'] \subseteq N[\toky]$ and thus as $v', y \in D_t$ we have that $D_t \setminus \{\toki'\}$ is dominating.
    Therefore, $\slide{D_t}{\toki'}{\toki}$ is dominating, a contradiction.
    See \Cref{fig:move2} for an illustration.

    Thus assume further that $\toki' \neq \scr_\toki$ and therefore $d(\toki, \toki') \geq 2$.
    Let $\scr^2_\toki \in \succr{\scr_\toki, \toki'}$.
    Note that $\scr^2_\toki$ must be adjacent to $\toky$ as $N[\scr_\toki] \subseteq N[\toky]$.
    See \Cref{fig:walk} for an illustration.
    We can construct a walk between $\toki'$ and $\toky'$ by concatenating shortest paths between each two consecutive vertices in $(\toky',\toky,\scr^2_\toki,\toki')$ of total length
    \[d(\toky', \toky) + 1 + d(\toki, \toki') - 2 = d(\toki, \toki') + d(\toky, \toky') - 1\]
    and therefore $d(\toki', \toky') \leq d(\toki, \toki') + d(\toky, \toky') - 1$.

    \subparagraph{Case 3: $\scr_\toki$ and $\scr_\toky$ are not adjacent and $I(\scr_\toki)$ is not nested in $I(\toky)$.}
    Recall that by Equation~\eqref{eq:ordering} it holds $\ell(\toky) < \ell(\scr_\toki)$.
    Furthermore, $r(\toky) < r(\scr_\toki)$, as otherwise $I(\scr_\toki)$ would be nested in $I(\toky)$.

    Let us now consider the possible orderings of the right endpoints of $I(\toki), I(y), I(\scr_\toki), I(\scr_\toky)$.
    The possibilities are restricted by the fact that by Equation~\eqref{eq:ordering} it holds $r(\toki) < r(\toky) < r(\scr_\toki)$, thus there remain $4$ possible orderings.
    The case $r(\toki) < r(\toky) < r(\scr_\toki) < r(\scr_\toky)$ can be ruled out as it contradicts $I(y)$ and $I(\scr_\toky)$ intersecting and $I(\scr_\toki)$ and $I(\scr_\toky)$ not intersecting at the same time.
    Similarly $r(\toki) < r(\toky) < r(\scr_\toky) < r(\scr_\toki)$ and $r(\toki) < r(\scr_\toky) < r(\toky) < r(\scr_\toki)$ is not possible as it would contradict $I(\toki)$ and $I(\scr_\toki)$ intersecting and at the same time $I(\scr_\toki)$ and $I(\scr_\toky)$ not intersecting.

    Thus, the only remaining ordering is $r(\scr_\toky) < r(\toki) < r(\toky) < r(\scr_\toki)$.
    This by \Cref{obs:sp} implies that either $\scr_{\toky} = \toky'$ or $r(\toky') < r(\toky)$.
    In either case, it follows that $r(\toky') < r(\toki)$ which contradicts the choice of $\toki$.
  \end{claimproof}
  We have shown that for any two dominating sets $D_s \neq D_t$ and a minimum-cost matching $M$ between them, we can construct another minimum matching $M'$ such that at least one of the following statements holds.
  There exists either $\toki \in D_s$ and $\toky \in D_t$ such that $(\toki, \toky) \in M'$ and $\slide{D_s}{\toki}{\toky}$ is dominating or,
  by a symmetric proof with $D_s$ and $D_t$ swapped, there exists $\toki \in D_t, \toky \in D_s$ such that $(\toky, \toki) \in M'$ and $\slide{D_t}{\toki}{\toky}$ is dominating.
  In either case, we have shown that $\toki$ and $\toky$ can be adjacent and thus $\toky \in \succr{\toki, \toky}$.
  Furthermore, $M'$ is constructed as described on~\ref{line:fix-matching} and $y$ can be found by testing all vertices in $D_t$.
  This concludes the proof.
\end{proof}

\begin{theorem}\label{thm:interval-graphs}
  Let $G$ be an interval graph with $n$ vertices and $D_s, D_t$ its two dominating sets such that $|D_s| = |D_t|$.
  Then \Call{ReconfIG}{} correctly computes a solution to \textsc{Shortest reconfiguration of dominating sets under \TS} in time $\calO(n^3)$, where $k$ is the size of the output.
\end{theorem}
\begin{proof}
We first show that the resulting reconfiguration sequence has the shortest possible length.
\begin{claim}
  The number of moves outputted by \Call{ReconfIG}{} is $d_\mathcal{R}(D_s, D_t)$
\end{claim}
\begin{claimproof}
  We will show that $d_\mathcal{R}(D_s, D_t) = c^*(D_s, D_t)$ by induction over $c^*(D_s, D_t)$.
  If $c^*(D_s, D_t) = 0$, then $D_s = D_t$ and $d_\mathcal{R}(D_s, D_t) = 0$, which we can efficiently recognize.

  Suppose that $c^*(D_s, D_t) > 0$.
  Let $M$ be a minimum-cost matching between $D_s$ and $D_t$.
  Without loss of generality, let $(u, v) \in M', u' \in \succr{u, v}$ such that $D'_s = \slide{D_s}{u}{u'}$ is dominating.
  By \Cref{lem:move}, either such $u, u', v$ already exist or we can recompute $M'$ so that they exist.

  Note that by \Cref{obs:matching-move}, $c^*(D'_s, D_t) = c^*(D_s, D_t) - 1$ and thus by the induction hypothesis $d_\mathcal{R}(D'_s, D_t) = c^*(D'_s, D_t)$.
  Note that $d_\mathcal{R}(D_s, D_t) \leq d_\mathcal{R}(D'_s, D_t) + 1$ as $D_s$ can be reached from $D_s$ by a single token slide.
  At the same time, by \Cref{lem:lb}, it holds $d_\mathcal{R}(D_s, D_t) \geq c^*(D_s, D_t) = c^*(D'_s, D_t) + 1 = d_\mathcal{R}(D'_s, D_t) + 1$.
  Thus $d_\mathcal{R}(D_s, D_t) = d_\mathcal{R}(D'_s, D_t)$ and with each output of a token slide, we decrease the distance in $d_\mathcal{R}$ by exactly one.
  Therefore, the resulting reconfiguration is shortest possible.
\end{claimproof}

\begin{claim}
  \Call{ReconfIG}{} can be implemented to run in time $\calO(n^3)$.
\end{claim}
\begin{claimproof}
  We initially compute a minimum-cost matching between $D_s$ and $D_t$ by reducing to minimum-cost matching in bipartite graphs, which can be solved in~$\mathcal{O}(n^3)$~\cite{Jonker_1987}.
  
  Now, we describe how to implement \Cref{alg:intervals} efficiently.
  If we want to find a suitable $v$ in \Call{FixMatching}{}, we suppose that all greedy moves, i.e. moves along shortest paths to matches that result in a dominating set, have been done.
  This is not necessary, we can see that the assumption is invoked only on constantly many vertices for each call of \Call{FixMatching}{}.
  Checking if a greedy move can be performed requires only linear time and the total number of moves is at most $\mathcal{O}(n^2)$, thus the total running time is $\mathcal{O}(n^3)$.
\end{claimproof}
\end{proof}

\section{Conclusion}

\begin{figure}
\begin{center}
  \begin{tikzpicture}
    \node[draw,circle] (a) at (0:0) {};
    \node[draw,circle] (a2) at (90:2) {};
    \node[draw,circle] (a5) at (270:2) {};

    \node[draw,blue,thick,fill] (a1) at (30:2) {};
    \node[draw,blue,thick,fill] (a4) at (210:2) {};
    \node[draw,red,thick,fill] (a3) at (150:2) {};
    \node[draw,red,thick,fill] (a6) at (330:2) {};
    \foreach \i in {1,2,...,6}{
      \draw (a)--(a\i);
    }
    
    \foreach \i in {1,2,...,5}{
      \pgfmathtruncatemacro{\y}{\i + 1};
      \draw (a\y)--(a\i);
    }
    \draw (a1)--(a6);

  \end{tikzpicture}
  \caption{Dually chordal graph where the lower bound from minimal matching is not achievable.
  The minimum-cost matching between the red and the blue vertices is $2$ but to reconfigure one into the other, we need at least $3$ moves.}\label{fig:dually-chordal}
\end{center}
\end{figure}
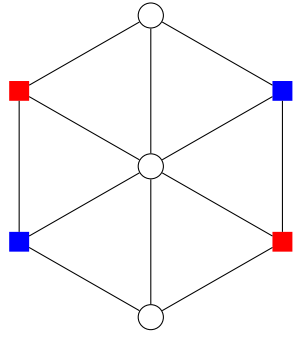

In this paper, we have presented polynomial algorithms for finding a shortest reconfiguration sequence between dominating sets on trees and interval graphs, addressing the open question left by Bonamy et al.~\cite{Bonamy_2021}.
Their work provided an efficient algorithm for finding a reconfiguration sequence between two dominating sets in dually chordal graphs, which include trees and interval graphs as subclasses.
We have shown that in case of trees and interval graphs, we can always match the lower bound of \Cref{lem:lb}.
That is not the case for dually chordal graph in general, see \Cref{fig:dually-chordal}.

While our work contributes to the understanding of reconfiguration problems in trees and interval graphs, the general case of dually chordal graphs remains open.
Additionally, the case of cographs is still open, and we conjecture that a polynomial-time solution is achievable.

It would be interesting to find a class of graphs for which in the case of dominating sets, the optimization variant is NP-hard while the reachability variant is polynomial-time solvable.
Furthermore, it would be intriguing to provide a polynomial-time algorithm for the optimization variant in a class of graphs that may require ``detour''.


\bibliographystyle{plainurl}
\bibliography{refs}


\end{document}